\documentclass[10pt,reqno]{amsart}
\usepackage{latexsym,amsmath,amssymb,amscd}

\def\today{\ifcase \month \or
   January \or February \or March \or April \or
   May \or June \or July \or August \or
   September \or October \or November \or December \fi
   \space\number\day , \number\year}
   
\makeatletter
  \newcommand\@dotsep{4.5}
  \def\@tocline#1#2#3#4#5#6#7{\relax
     \ifnum #1>\c@tocdepth 
     \else
     \par \addpenalty\@secpenalty\addvspace{#2}%
     \begingroup \hyphenpenalty\@M
     \@ifempty{#4}{%
     \@tempdima\csname r@tocindent\number#1\endcsname\relax
        }{%
         \@tempdima#4\relax
           }%
      \parindent\z@ \leftskip#3\relax \advance\leftskip\@tempdima\relax
      \rightskip\@pnumwidth plus1em \parfillskip-\@pnumwidth
       #5\leavevmode\hskip-\@tempdima #6\relax
       \leaders\hbox{$\m@th
       \mkern \@dotsep mu\hbox{.}\mkern \@dotsep mu$}\hfill
       \hbox to\@pnumwidth{\@tocpagenum{#7}}\par
       \nobreak
        \endgroup
         \fi}
\makeatother 

\begin{document}


\makeatletter
\@addtoreset{figure}{section}
\def\thefigure{\thesection.\@arabic\c@figure}
\def\fps@figure{h,t}
\@addtoreset{table}{bsection}

\def\thetable{\thesection.\@arabic\c@table}
\def\fps@table{h, t}
\@addtoreset{equation}{section}
\def\theequation{
\arabic{equation}}
\makeatother

\newcommand{\bfi}{\bfseries\itshape}

\newtheorem{theorem}{Theorem}
\newtheorem{acknowledgment}[theorem]{Acknowledgment}
\newtheorem{corollary}[theorem]{Corollary}
\newtheorem{definition}[theorem]{Definition}
\newtheorem{example}[theorem]{Example}
\newtheorem{lemma}[theorem]{Lemma}
\newtheorem{notation}[theorem]{Notation}
\newtheorem{proposition}[theorem]{Proposition}
\newtheorem{remark}[theorem]{Remark}
\newtheorem{setting}[theorem]{Setting}

\numberwithin{theorem}{section}
\numberwithin{equation}{section}

\newcommand{\1}{{\bf 1}}
\newcommand{\Ad}{{\rm Ad}}
\newcommand{\Alg}{{\rm Alg}\,}
\newcommand{\Aut}{{\rm Aut}\,}
\newcommand{\ad}{{\rm ad}}
\newcommand{\Borel}{{\rm Borel}}
\newcommand{\Ci}{{\mathcal C}^\infty}
\newcommand{\Cpol}{{\mathcal C}^\infty_{\rm pol}}
\newcommand{\Der}{{\rm Der}\,}
\newcommand{\de}{{\rm d}}
\newcommand{\ee}{{\rm e}}
\newcommand{\End}{{\rm End}\,}
\newcommand{\ev}{{\rm ev}}
\newcommand{\id}{{\rm id}}
\newcommand{\ie}{{\rm i}}
\newcommand{\GL}{{\rm GL}}
\newcommand{\gl}{{{\mathfrak g}{\mathfrak l}}}
\newcommand{\Hom}{{\rm Hom}\,}
\newcommand{\Img}{{\rm Im}\,}
\newcommand{\Ind}{{\rm Ind}}
\newcommand{\Ker}{{\rm Ker}\,}
\newcommand{\Lie}{\text{\bf L}}
\newcommand{\m}{\text{\bf m}}
\newcommand{\pr}{{\rm pr}}
\newcommand{\Ran}{{\rm Ran}\,}
\renewcommand{\Re}{{\rm Re}\,}
\newcommand{\so}{\text{so}}
\newcommand{\spa}{{\rm span}\,}
\newcommand{\Tr}{{\rm Tr}\,}
\newcommand{\Op}{{\rm Op}}
\newcommand{\U}{{\rm U}}

\newcommand{\CC}{{\mathbb C}}
\newcommand{\RR}{{\mathbb R}}
\newcommand{\TT}{{\mathbb T}}

\newcommand{\Ac}{{\mathcal A}}
\newcommand{\Bc}{{\mathcal B}}
\newcommand{\Cc}{{\mathcal C}}
\newcommand{\Dc}{{\mathcal D}}
\newcommand{\Ec}{{\mathcal E}}
\newcommand{\Fc}{{\mathcal F}}
\newcommand{\Hc}{{\mathcal H}}
\newcommand{\Jc}{{\mathcal J}}
\renewcommand{\Mc}{{\mathcal M}}
\newcommand{\Nc}{{\mathcal N}}
\newcommand{\Oc}{{\mathcal O}}
\newcommand{\Pc}{{\mathcal P}}
\newcommand{\Rc}{{\mathcal R}}
\newcommand{\Sc}{{\mathcal S}}
\newcommand{\Tc}{{\mathcal T}}
\newcommand{\Vc}{{\mathcal V}}
\newcommand{\Uc}{{\mathcal U}}
\newcommand{\Yc}{{\mathcal Y}}
\newcommand{\Wig}{{\mathcal W}}

\newcommand{\Bg}{{\mathfrak B}}
\newcommand{\Fg}{{\mathfrak F}}
\newcommand{\Gg}{{\mathfrak G}}
\newcommand{\Ig}{{\mathfrak I}}
\newcommand{\Jg}{{\mathfrak J}}
\newcommand{\Lg}{{\mathfrak L}}
\newcommand{\Pg}{{\mathfrak P}}
\newcommand{\Sg}{{\mathfrak S}}
\newcommand{\Xg}{{\mathfrak X}}
\newcommand{\Yg}{{\mathfrak Y}}
\newcommand{\Zg}{{\mathfrak Z}}

\newcommand{\ag}{{\mathfrak a}}
\newcommand{\bg}{{\mathfrak b}}
\newcommand{\dg}{{\mathfrak d}}
\renewcommand{\gg}{{\mathfrak g}}
\newcommand{\hg}{{\mathfrak h}}
\newcommand{\kg}{{\mathfrak k}}
\newcommand{\mg}{{\mathfrak m}}
\newcommand{\n}{{\mathfrak n}}
\newcommand{\og}{{\mathfrak o}}
\newcommand{\pg}{{\mathfrak p}}
\newcommand{\sg}{{\mathfrak s}}
\newcommand{\tg}{{\mathfrak t}}
\newcommand{\ug}{{\mathfrak u}}
\newcommand{\zg}{{\mathfrak z}}

\newcommand{\ZZ}{\mathbb Z}
\newcommand{\NN}{\mathbb N}
\newcommand{\BB}{\mathbb B}

\newcommand{\ep}{\varepsilon}

\newcommand{\hake}[1]{\langle #1 \rangle }

\newcommand{\scalar}[2]{\langle #1 ,#2 \rangle }
\newcommand{\vect}[2]{(#1_1 ,\ldots ,#1_{#2})}
\newcommand{\norm}[1]{\Vert #1 \Vert }
\newcommand{\normrum}[2]{{\norm {#1}}_{#2}}

\newcommand{\upp}[1]{^{(#1)}}
\newcommand{\p}{\partial}

\newcommand{\opn}{\operatorname}
\newcommand{\slim}{\operatornamewithlimits{s-lim\,}}
\newcommand{\sgn}{\operatorname{sgn}}

\newcommand{\seq}[2]{#1_1 ,\dots ,#1_{#2} }
\newcommand{\loc}{_{\opn{loc}}}

\makeatletter
\title[Uncertainty principles on coadjoint orbits]{Uncertainty principles for magnetic structures on 
certain coadjoint orbits}
\author{Ingrid Belti\c t\u a 
and Daniel Belti\c t\u a
}
\address{Institute of Mathematics ``Simion Stoilow'' 
of the Romanian Academy, 
P.O. Box 1-764, Bucharest, Romania}
\email{Ingrid.Beltita@imar.ro}
\email{Daniel.Beltita@imar.ro}
\keywords{Weyl calculus; magnetic field; 
Lie group; semidirect product}
\subjclass[2000]{Primary 81S30; Secondary 22E25, 22E65, 35S05, 47G30}
\date{June 8, 2009}
\makeatother

\begin{abstract} 
By building on our earlier work, 
we establish uncertainty principles in terms of Heisenberg inequalities 
and of the ambiguity functions 
associated with magnetic structures on 
certain coadjoint orbits of infinite-dim\-ensional Lie groups. 
These infinite-dimensional Lie groups are semidirect products 
of nilpotent Lie groups and invariant function spaces thereon. 
The recently developed magnetic Weyl calculus is recovered in the special case 
of function spaces on abelian Lie groups. 
\end{abstract}

\maketitle


\section{Introduction}

The relationship between the Weyl calculus of pseudo-differential operators on~$\RR^n$ 
and the Heisenberg group $\RR^{n+1}\rtimes\RR^n$ is a classical topic 
(see for instance \cite{Pe94}, \cite{Gr01}, or \cite{dG06}). 
In fact, the Weyl calculus provides a quantization of 
a nontrivial coadjoint orbit for the Heisenberg group. 
On the other hand, a magnetic gauge-invariant pseudo-differential calculus on $\RR^n$ 
has also been recently developed by using techniques of hard analysis; 
see \cite{MP04} and \cite{IMP07}. 
As our alternative approach has shown (\cite{BB09}), 
this magnetic calculus can be set up for any nilpotent Lie group~$G$ 
and can be understood as a quantization 
of a certain coadjoint orbit for some Lie group~$\Fc\rtimes G$, 
which is infinite dimensional unless 
the magnetic field is polynomial. 
More specifically, by adapting ideas of \cite{Ba98}, 
the cotangent bundle $T^*G$ 
has been symplectomorphically realized as a coadjoint orbit  
of $\Fc\rtimes G$ and the pseudo-differential calculus has been 
constructed as a Weyl quantization of that orbit. 
In our case, the semidirect product is needed in order to deal 
with rather general perturbations of invariant differential 
operators on~$G$. 
The semidirect products have also turned out to be an important 
tool in mechanics; see for instance \cite{HMR98}.

In the present paper we investigate some uncertainty principles 
for the magnetic Weyl calculus developed in \cite{BB09}. 
The uncertainty principles have been an active area of research.   
We refer to the survey \cite{FS97} for a comprehensive introduction 
to this circle of ideas, to \cite{HN88} for the case of families of pseudo-differential operators, and to \cite{Th04} and \cite{BK08} for 
Hardy's uncertainty principles on Lie groups.  
The main point of the present approach is 
that the aforementioned Weyl quantization
allows us to obtain versions of Heisenberg's inequality 
---taking into account magnetic momenta--- 
and Lieb's uncertainty principle (\cite{Li90}) for a certain wavelet transform 
associated with the coadjoint orbit $T^*G$ of~$\Fc\rtimes G$. 

Let us describe the contents of our paper in some more detail. 
Section~\ref{sect2} is devoted to establishing 
Heisenberg's uncertainty inequality in the magnetic setting 
on nilpotent Lie groups. 
In subsection~\ref{subsect2.1}, after describing 
the necessary notation used throughout the paper,  
we introduce the ambiguity function and the cross-Wigner distribution 
in the present framework and prove some of their main properties 
including Moyal's identity (Theorem~\ref{o6}). 
Preliminary material on magnetic Weyl calculus from~\cite{BB09} 
is provided in subsection~\ref{subsect2.1} 
along with additional properties in terms of the Wigner distribution. 
Thus, in Proposition~\ref{marginal} 
we indicate the significance of its marginal distributions for
the functional calculus with both the position operators 
and the ``noncommutative magnetic momentum'' operators. 
Let us point out that using the usual Fourier transform 
does not seem very natural in the present context. 
This is due both to the presence of the magnetic potential 
and to the fact that the invariant vector fields  
on a nilpotent Lie group may not have constant coefficients 
(see Example~\ref{part}). 
Versions for Heisenberg's inequality are established 
in Theorem~\ref{ineq_th} and Corollary~\ref{ineq_cor}. 

Section~\ref{sect3} deals with a version of Lieb's uncertainty principle 
in the present setting. 
The main result is Theorem~\ref{lieb} and is 
stated in terms of magnetic ambiguity functions and mixed-norm Lebesgue spaces 
on the cotangent bundle of a nilpotent Lie group. 
In the case of abelian Lie groups and no magnetic potential 
we recover one of the results of \cite{BDO07}. 
(See also \cite{BDJ03} and \cite{De05} for related results 
in this classical case.) 
Among the consequences of Theorem~\ref{lieb} 
we mention an embedding theorem for the natural versions 
of the modulation spaces in our setting (Corollary~\ref{lieb_cor1}). 

Finally, in Section~\ref{sect4} we illustrate the main ideas 
by considering the special case of two-step nilpotent Lie algebras.

\section{Heisenberg's uncertainty inequality in the magnetic setting on nilpotent Lie groups}\label{sect2}

\subsection{Moyal's identity on nilpotent Lie groups}\label{subsect2.1}
In this subsection we introduce the ambiguity function and the cross-Wigner distribution 
in the present setting and prove some of their main properties 
including Moyal's identity (Theorem~\ref{o6}). 
This property occurs in connection with a finite-dimensional coadjoint orbit 
of a semidirect product which is in general an infinite-dimensional Lie group 
(see Prop.~2.9 in \cite{BB09}). 
It corresponds to the orthogonality relations proved in \cite{Pe94} 
for the matrix coefficients 
of \emph{any} irreducible representation of a nilpotent Lie group. 
Let us also note that 
wavelet transforms associated with semidirect products of locally compact 
(or finite-dimensional Lie) groups appeared in \cite{KT03} and \cite{Fue09}.

\begin{setting}\label{o1}
\normalfont
We shall work in the setting of Section~4 in \cite{BB09}. 
Let us briefly recall the main notation involved therein. 
\begin{itemize}
\item A connected, simply connected, nilpotent Lie group $G$ is 
identified to its Lie algebra $\gg$ by means of the exponential map. 
We denote by $\ast$ the Baker-Campbell-Hausdorff multiplication 
on $\gg$, so that $G=(\gg,\ast)$. 
\item 
The cotangent bundle $T^*G$ is a trivial bundle and 
we perform the identification 
\begin{equation}\label{cot}
T^*G\simeq\gg\times\gg^*
\end{equation}
by using the trivialization by left translations. 
\item $\Fc$ is an admissible function space on the Lie group $G$ 
(see  Def.~2.8 in \cite{BB09});  
in particular, $\Fc$ is invariant under translations to the left on $G$ 
and is endowed with a locally convex topology 
such that we have continuous inclusions $\gg^*\hookrightarrow \Fc\hookrightarrow\Ci(G)$. 

For instance $\Fc$ can be the whole space $\Ci(G)$ or the space $\Cpol(G)$ 
of smooth functions with polynomial growth. 
See however Example~\ref{finite-dim} below for specific situations when $\dim\Fc<\infty$.
\item The semidirect product $M=\Fc\rtimes_\lambda G$ is an infinite-dimensional Lie group 
in general, whose Lie algebra is $\mg=\Fc\rtimes_{\dot\lambda}\gg$. 
We refer to \cite{Ne06} or \cite{Be06} for basic facts on infinite-dimensional Lie groups. 
\item We endow $\gg$ and its dual space $\gg^*$ with Lebesgue measures 
suitably normalized such that the Fourier transform $L^2(\gg)\to L^2(\gg)$ 
is a unitary operator, and we denote $\Hc=L^2(\gg)$. 
\item We define a unitary representation $\pi\colon M\to\Bc(\Hc)$ by 
$$(\pi(\phi,X)f)(Y)=\ee^{\ie\phi(Y)}f((-X)\ast Y)$$
for $(\phi,X)\in M$, $f\in\Hc$, and $Y\in\gg$. 
\item The magnetic potential is a smooth mapping $A\colon\gg\to\gg^*$, $X\mapsto A_X$, with polynomial growth 
such that for every $X\in\gg$ we have $\langle A_\bullet,(R_\bullet)'_0X\rangle\in\Fc$. 
\item We also need the mappings 
$$\theta_0\colon\gg\times\gg^*\to\Fc, \quad 
\theta_0(X,\xi)=\xi+\langle A_\bullet,(R_\bullet)'_0X\rangle$$
and 
$\theta\colon\gg\times\gg^*\to\mg$, $(X,\xi)\mapsto(\theta_0(X,\xi),X)$. 
Here $R_Y\colon\gg\to\gg$, $Z\mapsto Z\ast Y$, is the translation to the right 
defined by any $Y\in\gg$. 
\end{itemize}
\qed
\end{setting}

\begin{remark}\label{o2}
\normalfont 
For every $(X,\xi)\in\gg\times\gg^*$ we have 
$$(\pi(\exp_M(\theta(X,\xi)))f)(Y)=\exp\Bigl(\ie\int\limits_0^1
\theta_0(X,\xi)((-sX)\ast Y)\de s\Bigr)f((-X)\ast Y)$$
whenever $f\in L^2(\gg)$ and $Y\in\gg$. 
See eq.~(4.8) in \cite{BB09}. 
\qed
\end{remark}

\begin{notation}\label{o3}
\normalfont
We shall denote for every $X\in\gg$, 
$$\Psi_X\colon\gg\to\gg, \quad \Psi_X(Y)=\int\limits_0^1Y\ast(sX)\de s$$
(see Prop.~3.2 in \cite{BB09}) and also 
$$\tau_A(X,Y)=\exp\Bigl(\ie\int\limits_0^1\langle A_{(-sX)\ast Y},(R_{(-sX)\ast Y})'_0X\rangle\de s\Bigr)$$
for $X,Y\in\gg$. 
\qed
\end{notation}

\begin{lemma}\label{o4}
For every $(X,\xi)\in\gg\times\gg^*$ and $f\in L^2(\gg)$ 
we have 
\begin{equation}\label{o4_item1} 
(\pi(\exp_M(\theta(X,\xi)))f)(Y)
=\tau_A(X,Y)\ee^{-\ie\langle\xi,\Psi_X(-Y)\rangle}f((-X)\ast Y)
\end{equation} 
and 
\begin{equation}\label{o4_item2} 
(\pi(\exp_M(-\theta(X,\xi)))f)(Y)
=\tau_A(X,X\ast Y)^{-1}\ee^{\ie\langle\xi,\Psi_X(-(X\ast Y))\rangle}f(X\ast Y)
\end{equation}
for arbitrary $Y\in\gg$. 
\end{lemma}

\begin{proof}
Formula~\eqref{o4_item1} follows at once by Remark~\ref{o2} and Notation~\ref{o3}. 
In order to prove the second formula, note that for $\phi\in L^2(\gg)$ 
we have by~\eqref{o4_item1}
$$\begin{aligned}
(\pi(\exp_M(\theta(X,&\xi)))f)(Y)=\phi(Y) \\
&\iff \tau_A(X,Y)\exp(-\ie\langle\xi,\Psi_X(-Y)\rangle)f((-X)\ast Y)=\phi(Y) \\
&\iff f((-X)\ast Y)=\tau_A(X,Y)^{-1}\exp(\ie\langle\xi,\Psi_X(-Y)\rangle)\phi(Y) 
\end{aligned}$$
for arbitrary $Y\in\gg$, which is further equivalent to 
$$(\forall Y\in\gg)\quad f(Y)
=\tau_A(X,X\ast Y)^{-1}\exp(\ie\langle\xi,\Psi_X(-(X\ast Y))\rangle)\phi(X\ast Y) $$
and this concludes the proof. 
\end{proof}

\begin{definition}\label{o5}
\normalfont
For arbitrary $\phi,f\in L^2(\gg)$ we define the function 
$${\Ac}_\phi f\colon\gg\times\gg^*\to\CC,\quad 
({\Ac}_\phi f)(X,\xi)=(f\mid\pi(\exp_M(\theta(X,\xi)))\phi). $$
We shall call $\Ac_\phi f$ 
the \emph{ambiguity function} 
defined by $\phi,f\in L^2(\gg)$. 
By using the canonical symplectic structure on $\gg\times\gg^*$ given by 
$$(\gg\times\gg^*)\times(\gg\times\gg^*)\to\RR, 
\quad ((X_1,\xi_1),(X_2,\xi_2))\mapsto\hake{\xi_1,X_2}-\hake{\xi_2,X_1} $$
we also define the symplectic Fourier transform of the ambiguity function 
$$\Wig(f,\phi):=\widehat{\Ac_\phi f}\in L^2(\gg\times\gg^*)$$ 
and we call it the \emph{cross-Wigner distribution (function)} of $\phi,f\in L^2(\gg)$. 
The definition of $\Wig(f,\phi)$ makes sense since 
it follows by Theorem~\ref{o6} below that $\Ac_\phi f\in L^2(\gg)$. 
\qed
\end{definition}

\begin{remark}\label{ext_amb}
\normalfont
Let $\phi\in\Sc(\gg)$. 
Formula~\eqref{o4_item1} shows that 
for every $(X,\xi)\in\gg\times\gg^*$ we have $\pi(\exp_M(\theta(X,\xi)))\phi\in\Sc(\gg)$. 
Moreover, the mapping 
$$\gg\times\gg^*\to\Sc(\gg),\quad (X,\xi)\mapsto\pi(\exp_M(\theta(X,\xi)))\phi$$
is continuous. 
Thus we can extend the definition of $\Ac_\phi f$ for every $f\in\Sc'(\gg)$ 
to obtain the continuous function 
$$\Ac_\phi f\colon\gg\times\gg^*\to\CC,\quad 
(\Ac_\phi f)(X,\xi)=\hake{f,\overline{\pi(\exp_M(\theta(X,\xi)))\phi}},$$ 
where $\hake{\cdot,\cdot}\colon\Sc'(\gg)\times\Sc(\gg)\to\CC$ 
is the usual duality pairing. 

We also note that if $f,\phi\in\Sc(\gg)$, then $\Ac_\phi f\in\Sc(\gg\times\gg^*)$ as 
an easy consequence of Lemma~\ref{o4}. 
\qed
\end{remark}

The second equality in Theorem~\ref{o6}\eqref{o6_item1} below 
will be referred to as \emph{Moyal's identity} just as 
in the classical situation when the Lie algebra $\gg$ is abelian 
(see for instance \cite{Gr01}).

\begin{theorem}\label{o6}
The following assertions hold: 
\begin{enumerate}
\item\label{o6_item1} 
For every $\phi,f\in L^2(\gg)$ we have ${\Ac}_\phi f\in L^2(\gg\times\gg^*)$ 
and 
\begin{equation*}
\begin{aligned}
({\Ac}_{\phi_1}f_1\mid {\Ac}_{\phi_2}f_2)_{L^2(\gg\times\gg^*)} 
&=(f_1\mid f_2)_{L^2(\gg)}\cdot(\phi_2\mid\phi_1)_{L^2(\gg)} \\
&=({\Wig}(f_1,\phi_1)\mid {\Wig}(f_2,\phi_2))_{L^2(\gg\times\gg^*)} 
\end{aligned}
\end{equation*}
whenever $\phi_1,f_1,\phi_2,f_2\in L^2(\gg)$. 
\item\label{o6_item2} 
If $\phi_0\in L^2(\gg)$ with $\Vert\phi_0\Vert=1$, 
then the operator ${\Ac}_{\phi_0}\colon L^2(\gg)\to L^2(\gg\times\gg^*)$, $f\mapsto {\Ac}_{\phi_0} f$, 
is an isometry and we have 
\begin{equation*}
\iint\limits_{\gg\times\gg^*}({\Ac}_{\phi_0}f)(X,\xi)\cdot\pi(\exp_M(\theta(X,\xi)))\phi\,\de(X,\xi)
=(\phi\mid\phi_0)f
\end{equation*}
for every $\phi,f\in L^2(\gg)$. 
In particular, 
\begin{equation*}
\iint\limits_{\gg\times\gg^*}({\Ac}_{\phi_0} f)(X,\xi)\cdot\pi(\exp_M(\theta(X,\xi)))\phi_0\,\de(X,\xi)=f
\end{equation*}
for arbitrary $f\in L^2(\gg)$. 
\end{enumerate}
\end{theorem}

\begin{proof}
\eqref{o6_item1} 
We may assume $f_1,f_2,\phi_1,\phi_2\in\Sc(\gg)$. 
Let $X\in\gg$ be fixed for the moment. 
We have by Lemma~\ref{o4}
$$
\begin{aligned}
\int\limits_{\gg^*} {\Ac}_{\phi_1}f_1(X,\xi)&\cdot\overline{{\Ac}_{\phi_2}f_2(X,\xi)}\de\xi \\
=&\lim_{\varepsilon\to0}\int\limits_{\gg^*}\ee^{-\epsilon\vert\xi\vert^2}\cdot
 {\Ac}_{\phi_1}f_1(X,\xi)\cdot\overline{{\Ac}_{\phi_2}f_2(X,\xi)}\de\xi \\
=&\lim_{\varepsilon\to0}\iiint\limits_{\gg^*\times\gg\times\gg}\ee^{-\epsilon\vert\xi\vert^2}\cdot
f_1(Y_1)\cdot\overline{\tau_A(X,Y_1)}\cdot\overline{\phi_1((-X)\ast Y_1)}\cdot\overline{f_2(Y_2)} \\
&\times \tau_A(X,Y_2)\cdot\phi_1((-X)\ast Y_2)\cdot
\ee^{\ie\langle\xi,\Psi_X(-Y_1)-\Psi_X(-Y_2)\rangle}
\de Y_1\de Y_2\de\xi \\
=&\lim_{\varepsilon\to0}\langle U_\varepsilon,F_X\rangle
\end{aligned}
$$
where $\langle\cdot,\cdot\rangle\colon\Sc'(\gg\times\gg)\times\Sc(\gg\times\gg)\to\CC$ 
stands for the usual duality between the tempered distributions and the Schwartz space.  
Here we think of the function 
$$U_\varepsilon(Y_1,Y_2)=\int\limits_{\gg^*}\ee^{-\epsilon\vert\xi\vert^2}\cdot
 \ee^{\ie\langle\xi,\Psi_X(-Y_1)-\Psi_X(-Y_2)\rangle}\de\xi$$
as tempered distribution on $\gg\times\gg$, while the function 
$$F_X(Y_1,Y_2)=f_1(Y_1)\cdot\overline{\tau_A(X,Y_1)}
\cdot\overline{\phi_1((-X)\ast Y_1)}\cdot\overline{f_2(Y_2)}\cdot
\tau_A(X,Y_2)\cdot\phi_1((-X)\ast Y_2) $$
belongs to $\Sc(\gg\times\gg)$. 
Since $\Psi_X\colon\gg\to\gg$ is a polynomial diffeomorphism of $\gg$ whose inverse is again 
a polynomial diffeomorphism (by Prop.~3.2 in \cite{BB09}), it follows by a standard reasoning 
that 
$$\lim_{\varepsilon\to0}U_\varepsilon
=\frac{1}{(2\pi)^n}\delta(\Psi_X(-Y_1)-\Psi_X(-Y_2))
=\frac{1}{(2\pi)^n}\delta(Y_1-Y_2) $$
in the weak topology of the space~$\Sc'(\gg\times\gg)$, 
where $\delta(\cdot)$ is the Dirac distribution 
at~$0\in\gg$. 
We then obtain 
$$
\begin{aligned}
\int\limits_{\gg^*} {\Ac}_{\phi_1}f_1(X,\xi)&\cdot\overline{{\Ac}_{\phi_2}f_2(X,\xi)}\de\xi \\
=& \frac{1}{(2\pi)^n}\int\limits_{\gg}F_X(Y,Y)\,\de Y \\
=& \frac{1}{(2\pi)^n}\int\limits_{\gg}f_1(Y)\cdot\overline{\phi_1((-X)\ast Y)}\cdot
f_2(Y)\cdot\phi_2((-X)\ast Y)\,\de Y
\end{aligned}
$$
since $\vert\tau_A(X,Y)\vert=1$. 
By integrating the above equality with respect to $X\in\gg$ 
and taking into account our convention on the relationship between 
the Lebesgue measures on $\gg$ and $\gg^*$, we eventually get 
$$({\Ac}_{\phi_1}f_1\mid {\Ac}_{\phi_2}f_2)_{L^2(\gg\times\gg^*)} 
=(f_1\mid f_2)_{L^2(\gg)}\cdot(\phi_2\mid\phi_1)_{L^2(\gg)}.$$
This is just the first equation we wished for. 
The second equality in the assertion follows from this one 
by using the well-known fact that the symplectic Fourier transform 
$L^2(\gg\times\gg^*)\to L^2(\gg\times\gg^*)$ 
is a unitary operator. 

\eqref{o6_item2} 
It follows at once by Assertion~\eqref{o6_item1} that 
the operator ${\Ac}_{\phi_0}\colon L^2(\gg)\to L^2(\gg\times\gg^*)$ 
is an isometry if $\Vert\phi_0\Vert=1$. 
The other properties then follow by general arguments; 
see for instance Proposition~2.11 in \cite{Fue05}. 
\end{proof}

\begin{proposition}\label{o7}
If $f,\phi\in\Sc(\gg)$, then the following assertions hold: 
\begin{enumerate}
\item\label{o7_item1} 
For every $(X,\xi)\in\gg\times\gg^*$ we have 
$$\begin{aligned}
({\Ac}_\phi f)(X,\xi)=\int\limits_{\gg} & \ee^{-\ie\hake{\xi,Y}}\overline{\tau_A(X,-\Psi_X^{-1}(-Y))} \\
& \times f(-\Psi_X^{-1}(-Y))\overline{\phi((-X)\ast (-\Psi_X^{-1}(-Y)))}\,\de Y. 
\end{aligned}$$
\item\label{o7_item2} 
For every $(Y,\eta)\in\gg\times\gg^*$ we have 
$$\begin{aligned}
\Wig(f,\phi)(Y,\eta)=\int\limits_{\gg} & \ee^{-\ie\hake{\eta,X}}\overline{\tau_A(X,-\Psi_X^{-1}(-Y))} \\
& \times f(-\Psi_X^{-1}(-Y))\overline{\phi((-X)\ast (-\Psi_X^{-1}(-Y)))}\,\de X. 
\end{aligned}$$
\end{enumerate}
\end{proposition}

\begin{proof}
It follows by Definition~\ref{o5} and Lemma~\ref{o4} that 
$$({\Ac}_\phi f)(X,\xi)=\int\limits_{\gg}f(Z)\overline{\tau_A(X,Z)}\ee^{\ie\hake{\xi,\Psi_X(-Z)}}
\overline{\phi((-X)\ast Z)}\,\de Z. 
$$
Since $\Psi_X\colon\gg\to\gg$ is a diffeomorphism with 
the Jacobian function equal to~1 everywhere, we can change variables and 
set $Y=-\Psi_X(-Z)$ in the above integral. 
Then $Z=-\Psi_X^{-1}(-Y)$ and we get the formula in Assertion~\eqref{o7_item1}. 
Then recall from Definition~\ref{o5} that 
$$\Wig(f,\phi)(Y,\eta)=\iint\limits_{\gg\times\gg^*}\ee^{-\ie(\hake{\eta,X}-\hake{\xi,Y})}
(\Ac_\phi f)(X,\xi)\,\de\xi\de X $$
If we plug in the formula of Assertion~\eqref{o7_item1} in the above equation 
and use the Fourier inversion formula, then we get the formula for $\Wig(f,\phi)(Y,\eta)$ 
as claimed. 
\end{proof}

\begin{remark}\label{o8}
\normalfont
It follows by Proposition~\ref{o7}\eqref{o7_item1} 
that the function ${\Ac}_\phi f(X,\cdot)\colon\gg^*\to\CC$ 
is equal to the inverse Fourier transform of the function 
$$\overline{\tau_A(X,-\Psi_X^{-1}(\cdot))}f(-\Psi_X^{-1}(\cdot))\overline{\phi((-X)\ast (-\Psi_X^{-1}(\cdot)))}\colon\gg\to\CC.$$
\qed
\end{remark}

\begin{remark}\label{o9}
\normalfont
We can use the above Proposition~\ref{o7} along with Prop.~3.2 in \cite{BB09} 
to check that the bilinear mappings 
$$\Ac(\cdot,\cdot),\Wig(\cdot,\cdot)\colon\Sc(\gg)\times\Sc(\gg)\to\Sc(\gg\times\gg^*)$$
are continuous. 
\qed
\end{remark}

\subsection{Magnetic pseudo-differential operators 
and Wigner distributions}\label{subsect2.2}
This subsection includes background material 
from~\cite{BB09} together with some new properties 
of the magnetic Weyl calculus on nilpotent Lie groups. 

\begin{definition}\label{pseudo}
\normalfont  
For every $a\in\Sc(\gg\times\gg^*)$ 
the corresponding \emph{magnetic pseudo-differential operator} is defined by 
\begin{equation}\label{pseudo_def}
\Op(a)f=\iint\limits_{\gg\times\gg^*}\widehat{a}(X,\xi)\cdot\pi(\exp_M(\theta(X,\xi)))f\,\de(X,\xi)
\end{equation}
for every $f\in\Sc(\gg)$, where $\theta\colon\gg\times\gg^*\to\Lie(M)$ is 
described in Setting~\ref{o1}
\qed
\end{definition}

We record in the following proposition some immediate properties 
of the magnetic pseudo-differential operators constructed 
in Definition~\ref{pseudo}. 

\begin{proposition}\label{pseudo_prop}
The following assertions hold: 
\begin{enumerate}
\item\label{pseudo_prop_item2} 
For each $a\in\Sc(\gg\times\gg^*)$ we have 
$$(\Op(a)f\mid\phi)_{L^2(\gg)}=(\widehat{a}\mid\Ac_f\phi)_{L^2(\gg\times\gg^*)}
=(a\mid\Wig(\phi,f))_{L^2(\gg\times\gg^*)}$$
whenever $f,\phi\in\Sc(\gg)$. 
\item\label{pseudo_prop_item3} 
If $\phi_1,\phi_2\in\Sc(\gg)$ and $a:=\Wig(\phi_1,\phi_2)\in\Sc(\gg\times\gg^*)$, 
then $\Op(a)$ is a rank-one operator, namely 
$$\Op(a)f=(f\mid\phi_2)_{L^2(\gg)}\cdot \phi_1 \text{ for every }f\in\Sc(\gg).$$ 
\end{enumerate}
\end{proposition}

\begin{proof}
Assertion~\eqref{pseudo_prop_item2} is a consequence of formula~\eqref{pseudo_def} 
along with Definition~\ref{o5}. 
Then Assertion~\eqref{pseudo_prop_item3} follows by Assertion~\eqref{pseudo_prop_item2} 
by taking into account Moyal's identity (Theorem~\ref{o6}\eqref{o6_item1}). 
In fact, we get 
$$\begin{aligned}
(\Op(\Wig(\phi_1,\phi_2))f\mid\phi)
&=(\Wig(\phi_1,\phi_2)\mid\Wig(\phi,f))
=(\phi_1\mid\phi)\cdot(f\mid\phi_2) \\
&=((f\mid\phi_2)\phi_1 \mid\phi) 
\end{aligned}$$
for arbitrary $\phi\in\Sc(\gg)$, 
and the conclusion follows since $\Sc(\gg)$ is dense in $L^2(\gg)$. 
\end{proof}

\begin{remark}\label{pseudo_ext}
\normalfont
We can use the equations in above Proposition~\ref{pseudo_prop}\eqref{pseudo_prop_item2} 
and Remark~\ref{o9} 
to define for every $a\in\Sc'(\gg\times\gg^*)$ the corresponding 
magnetic pseudo-differential operator as a continuous linear operator 
$\Op(a)\colon\Sc(\gg)\to\Sc'(\gg)$. 
It follows by this definition that the following assertions hold: 
\begin{enumerate}
\item If $\lim\limits_{j\in J}a_j=a$ in the weak$^*$-topology in $\Sc'(\gg\times\gg^*)$, 
then for every $f\in\Sc(\gg)$ we have 
$\lim\limits_{j\in J}\Op(a_j)f=\Op(a)f$ in the weak$^*$-topology in $\Sc'(\gg)$. 
\item The distribution kernel $K_a\in\Sc'(\gg\times\gg)$ of the operator 
$\Op(a)\colon\Sc(\gg)\to\Sc'(\gg)$ is given by the formula 
\begin{equation}\label{kernel1}
K_a=\alpha_A\cdot (((1\otimes F_{\gg}^{-1})a)\circ\Sigma),
\end{equation}
where the function $\alpha_A$ multiplies 
the composition between partial inverse Fourier transform 
$(1\otimes F_{\gg}^{-1})a\in\Sc'(\gg\times\gg)$ 
and the polynomial diffeomorphism 
\begin{equation*}
\Sigma\colon\gg\times\gg\to\gg\times\gg,\quad 
\Sigma(X,Y)=\Bigl( \int \limits_0^1 (s(Y\ast (-X))) \ast X) \, \de s, X\ast (-Y) \Bigr)
\end{equation*}
whose inverse is again polynomial. 
In fact, this follows by Th.~4.4 and eq.~(4.14) in \cite{BB09} 
for $a\in\Sc(\gg\times\gg^*)$. 
Then the general case 
can be obtained by the preceding continuiy property, 
since $\Sc(\gg\times\gg^*)$ is weakly$^*$-dense in $\Sc'(\gg\times\gg^*)$. 
\end{enumerate}
For the sake of completeness, let us write~\eqref{kernel1} 
explicitly as 
\begin{equation}\label{kernel2}
K_a(X,Y)=\alpha_A(X,Y)\int\limits_{\gg^*}
\ee^{\ie\hake{\xi,X\ast (-Y)}}a\Bigl( \int \limits_0^1 (s(Y\ast (-X))) \ast X \, \de s,\xi \Bigr)\,\de\xi
\end{equation}
which makes sense whenever $a\in\Sc'(\gg\times\gg^*)$ 
is defined by a function such that the right-hand side is well defined. 
Here we have used the notation 
\begin{equation}\label{alpha_A}
\alpha_A(X,Y) = 
\exp\Bigl({\ie \int\limits_0^1 \scalar{A((s(Y\ast(-X)))\ast X)}{(R_{(s(Y\ast(-X)))\ast X})'_0(X\ast (-Y))} \, \de s}\Bigr)
\end{equation} 
for every $X,Y\in\gg$ 
(see eq.~(4.13) in \cite{BB09}). 
\qed
\end{remark}

\begin{example}\label{pseudo_ex}
\normalfont
We wish to use Remark~\ref{pseudo_ext} in order to compute 
the magnetic pseudo-differential operators defined by some special types of symbols. 
\begin{enumerate}
\item\label{pseudo_ex_item1} 
Let $a\colon\gg\to\CC$ be a smooth function of polynomial growth 
and look at it as a symbol in $\Sc'(\gg\times\gg^*)$ 
depending only on the variable in~$\gg$. 
Since $\alpha_A(X,X)=1$, it the follows at once from~\eqref{kernel1} 
that $\Op(a)$ is the multiplication operator in $L^2(\gg)$ defined by 
the function~$a$. 
\item\label{pseudo_ex_item2}
Let $X_0\in\gg$ and define $a_{X_0}\colon\gg\times\gg^*\to\CC$, $a_{X_0}(X,\xi)=\hake{\xi,X_0}$. 
Then it follows by Th.~4.4(1) (and its proof) in \cite{BB09} that 
$$\Op(a_{X_0})=-\ie\dot{\lambda}(X_0)+A(Q)X_0$$
and this operator is the infinitesimal generator of a 1-parameter group of unitary operators, 
hence it is essentially self-adjoint in~$L^2(\gg)$. 
Here $\dot{\lambda}(X_0)$ is the first-order differential operator 
defined by the right-invariant vector field~$\overline{X}_0$ on the nilpotent Lie group $(\gg,\ast)$ 
whose value at $0\in\gg$ is~$X_0$. 
On the other hand, $A(Q)X_0$ stands for the multiplication operator 
given by the function whose value at an arbitrary point 
is obtained by applying the 1-form $A\in\Omega^1(\gg)$ 
to the aforementioned vector field~$\overline{X}_0$. 
Let us note that an explicit formula for $\dot{\lambda}(X_0)$ 
can be easily obtained by Lemma~5 in \cite{Ma07}, namely 
for every $f\in\Ci(\gg)$ and $Y\in\gg$ we have 
$(\dot{\lambda}(X_0)f)(Y)=\hake{f'_Y,\overline{X}_0(Y)}$, 
which is the first-order differential operator defined by the vector field 
$\overline{X}_0\colon\gg\to\gg$, 
\begin{equation}\label{maillard}
\overline{X}_0(Y)=\Rc(\ad_\gg Y)X_0=X_0-\frac{1}{2}[Y,X_0]+\frac{1}{12}[Y,[Y,X_0]]+\cdots. 
\end{equation}
Here we use the holomorphic function 
$\Rc\colon\CC\setminus 2\pi\ie\ZZ^*\to\CC$, $\Rc(z)=z/(e^z-1)$ 
whose power series around $0$ is $1-\frac{1}{2}z+\frac{1}{12}z^2+\cdots$. 
\item\label{pseudo_ex_item3} 
Now assume that the magnetic potential $A$ vanishes.  
Let $a\in L^1(\gg^*)$ and think of it as a symbol in $\Sc'(\gg\times\gg^*)$ 
depending only on the variable in~$\gg^*$. 
If we denote by $b\in L^\infty(\gg)$ the inverse Fourier transform of~$a$, 
then it follows by~\eqref{kernel2} that $K_a(X,Y)=b(X\ast(-Y))$, hence 
$$(\forall f\in\Sc(\gg)) \quad (\Op(a)f)(X)=\int\limits_{\gg}b(X\ast(-Y))f(Y)\,\de X. $$
Thus $\Op(a)$ is a convolution operator on the nilpotent Lie group $(\gg,\ast)$. 
\end{enumerate}
\qed
\end{example}

Our next aim is to show that the Weyl calculus with real symbols 
gives rise to symmetric pseudo-differential operators; 
see Proposition~\ref{self-adj} below. 

\begin{lemma}\label{symm}
If we define
$$\Sigma_1\colon\gg\times\gg\to\gg,\quad 
\Sigma_1(X,Y)=\int\limits_0^1(s(Y\ast (-X)))\ast X\,\de s, $$
then for every $X,Y\in\gg$ we have $\Sigma_1(X,Y)=\Sigma_1(Y,X)$. 
\end{lemma}

\begin{proof}
Note that for every $X,Y\in\gg$ we have 
\begin{equation*}
\Psi_X(Y\ast(-X))=\int\limits_0^1 Y\ast(-X)\ast sX \,\de s
=\int\limits_0^1 Y\ast(-(1-s)X)\,\de s
=\Psi_{-X}(Y).
\end{equation*}
If we replace $X$ by $(-X)\ast Y$, 
then we get $\Psi_{(-X)\ast Y}(Y\ast(-Y)\ast X)=\Psi_{(-Y)\ast X}(Y)$, 
that is, 
$$(\forall X,Y\in\gg)\quad \Psi_{(-X)\ast Y}(X)=\Psi_{(-Y)\ast X}(Y).$$
Now the conclusion follows since 
$$\Sigma_1(X,Y)=-\Psi_{X\ast(-Y)}(-X)$$ 
for every $X,Y\in\gg$. 
\end{proof}

\begin{proposition}\label{self-adj}
Let $a\in\Sc'(\gg\times\gg^*)$ be a real distribution, 
in the sense that its values on real valued functions are real numbers. 
Then the distribution kernel $K_a\in\Sc'(\gg\times\gg)$ has 
the following symmetry property: 
$$(\forall f,\phi\in\Sc(\gg))\quad 
\hake{K_a,f\otimes\bar{\phi}}=\overline{\hake{K_a,\phi\otimes\bar{f}}}. $$
\end{proposition}

\begin{proof}
Firstly note that for every $X,Y\in\gg$ we have by \eqref{alpha_A}
\allowdisplaybreaks
\begin{align}
\alpha_A(Y,X) 
&=\exp\Bigl({\ie\int\limits_0^1\scalar{A((s(X\ast(-Y)))\ast Y)}{(R_{(s(X\ast(-Y)))\ast Y})'_0(Y\ast (-X))} \, \de s}\Bigr) \nonumber\\
&=\exp\Bigl({\ie\int\limits_0^1\scalar{A((-sZ)\ast Z\ast X)}{(R_{(-sZ)\ast Z\ast X})'_0 Z} \, \de s}\Bigr) \nonumber\\
&=\exp\Bigl({\ie\int\limits_0^1\scalar{A(((1-s)Z)\ast X)}{(R_{((1-s)Z)\ast X})'_0 Z} \, \de s}\Bigr) \nonumber\\
&=\exp\Bigl({\ie\int\limits_0^1\scalar{A((sZ)\ast X)}{(R_{(sZ)\ast X})'_0 Z} \, \de s}\Bigr) \nonumber\\
&=\exp\Bigl({\ie\int\limits_0^1\scalar{A((s(Y\ast(-X)))\ast X)}{(R_{(s(Y\ast(-X)))\ast X})'_0((Y\ast(-X)))} \, \de s}\Bigr) \nonumber\\
&=\exp\Bigl({-\ie\int\limits_0^1\scalar{A((s(Y\ast(-X)))\ast X)}{(R_{(s(Y\ast(-X)))\ast X})'_0((X\ast(-Y)))} \, \de s}\Bigr) \nonumber\\
&=\overline{\alpha(X,Y)}, \nonumber
\end{align}
where we used the notation $X\ast(-Y)=-Z$, hence $Y=Z\ast X$. 
Now the assertion follows at once by using Lemma~\ref{symm} and 
formula~\eqref{kernel2}. 
\end{proof}

The next result shows the significance of the marginal distributions 
of the cross-Wigner function in our setting. 
It is worth pointing out that this is a natural extension of 
the similar property in the classical case of the Schr\"odinger representation.
Actually, the functional calculus with both the position operators 
(see Assertion~\eqref{marginal_item1}) 
and the ``noncommutative magnetic momentum'' operators $-\ie\dot{\lambda}(X_0)+A(Q)X_0$
(Assertion~\eqref{marginal_item2}) 
can thus be read off with the cross-Wigner distribution. 

\begin{proposition}\label{marginal}
If $f,\phi\in\Sc(\gg)$, then the following assertions hold: 
\begin{enumerate}
\item\label{marginal_item1} 
For every $Y\in\gg$ we have 
$$f(Y)\overline{\phi(Y)}=\int\limits_{\gg^*}\Wig(f,\phi)(Y,\eta)\,\de\eta. $$
\item\label{marginal_item2}
If we define 
$$\Gamma_{f,\phi}\colon\gg^*\to\CC,\quad 
\Gamma_{f,\phi}(\eta)=\int\limits_{\gg}\Wig(f,\phi)(Y,\eta)\,\de Y,$$
then for every $X_0\in\gg$ and $a_0\in\Sc(\RR)$ we have 
$$(a_0(-\ie\dot{\lambda}(X_0)+A(Q)X_0)f\mid\phi)
=\int\limits_{\gg^*}\Gamma_{f,\phi}(\eta)a_0(\hake{\eta,X_0})\,\de\eta, $$
where the left-hand side involves the Borel functional calculus 
for the essentially self-adjoint operator $-\ie\dot{\lambda}(X_0)+A(Q)X_0$ 
in~$L^2(\gg)$. 
\end{enumerate}
\end{proposition}

\begin{proof}
For Assertion~\eqref{marginal_item1} use Proposition~\ref{o7}\eqref{o7_item2} 
along with the Fourier inversion formula to get 
$$\begin{aligned}
\int\limits_{\gg^*}\Wig(f,\phi)(Y,\eta)\,\de\eta
&=\overline{\tau_A(0,-\Psi_0^{-1}(-Y))} 
\cdot f(-\Psi_0^{-1}(-Y))\cdot \overline{\phi(0\ast (-\Psi_0^{-1}(-Y)))} \\
&=f(Y)\overline{\phi(Y)}. 
\end{aligned}$$
The latter equality follows at once by the formulas in Notation~\ref{o3}. 

In order to prove Assertion~\eqref{marginal_item2}, let us denote by 
$\1\otimes a_0(\hake{\cdot,X_0})$ the function defined on $\gg\times\gg^*$ by 
$(X,\xi)\mapsto a_0(\hake{\xi,X_0})$. 
It then follows by Example~\ref{pseudo_ex}\eqref{pseudo_ex_item2} 
that $\1\otimes a_0(\hake{\cdot,X_0})=a_0\circ a_{X_0}$ 
and then 
$$a_0(-\ie\dot{\lambda}(X_0)+A(Q)X_0)=\Op(\1\otimes a_0(\hake{\cdot,X_0}))$$
(see also Sect.~5.1 in~\cite{Be06}). 
By using this equality along with Remark~\ref{pseudo_ext} and the formula 
in Proposition~\ref{pseudo_prop}\eqref{pseudo_prop_item2}, 
we get 
$$\begin{aligned}
(a_0(-\ie\dot{\lambda}(X_0)+A(Q)X_0)f\mid\phi)_{L^2(\gg)} 
&=(\1\otimes a_0(\hake{\cdot,X_0})\mid\Wig(\phi,f))_{L^2(\gg\times\gg^*)} \\
&=\iint\limits_{\gg\times\gg^*}a_0(\hake{\eta,X_0})\cdot\Wig(f,\phi)(Y,\eta)\,\de Y\de\eta \\
&=\int\limits_{\gg^*}a_0(\hake{\eta,X_0})\Bigl(\int\limits_{\gg}\Wig(f,\phi)(Y,\eta)\,\de Y\Bigr)\de\eta, 
\end{aligned}$$
and this leads to the asserted formula. 
\end{proof}

\subsection{Heisenberg's inequality}\label{subsect2.3}

In the following statement we shall use the symbols 
\begin{eqnarray}
a_{X_0}\colon\gg\times\gg^*\to\CC, & & a_{X_0}(X,\xi)=\hake{\xi,X_0}, \nonumber \\
a_{\xi_0}\colon\gg\times\gg^*\to\CC, & & a_{\xi_0}(X,\xi)=\hake{\xi_0,X}. \nonumber
\end{eqnarray}
for arbitrary $X_0\in\gg$ and $\xi_0\in\gg^*$. 

\begin{theorem}\label{ineq_th}
Let $X_0\in\gg$ and $c_0\in\RR$. 
Assume that the coadjoint orbit $\Oc\subseteq\gg^*$ is contained in the affine hyperplane 
$\{\xi\in\gg^*\mid\hake{\xi,X_0}=c_0\}$. 
Then  
\begin{equation}\label{ineq_th_eq1}
[\Op(a_{X_0}),\Op(a_{\xi_0})]=\ie c_0\cdot\id_{L^2(\gg)} 
\end{equation}
and 
\begin{equation}\label{ineq_th_eq2}
\Vert\Op(a_{X_0})f\Vert\cdot\Vert\Op(a_{\xi_0})f\Vert\ge\frac{1}{2}\vert c_0\vert 
\end{equation} 
for every $\xi_0\in\Oc$, 
whenever $f\in L^2(\gg)$ with $\Vert f\Vert=1$ belongs to the domains of both operators 
$\Op(a_{X_0})$ and $\Op(a_{\xi_0})$. 
\end{theorem}

\begin{proof}
For the sake of simplicity we shall use the convention that 
the operator of multiplication by some function 
will be denoted by the same symbol as that function. 
Then, according to Example~\ref{pseudo_ex}\eqref{pseudo_ex_item1}--\eqref{pseudo_ex_item2} 
we have $\Op(a_{\xi_0})=\xi_0$ 
and $\Op(a_{X_0})=-\ie\dot{\lambda}(X_0)+A(Q)X_0$. 
Since $\ie\dot{\lambda}(X_0)$ is a first-order linear differential operator on $\Ci(\gg)$, 
hence a derivation on $\Ci(\gg)$, it easily follows that 
\begin{equation}\label{ineq_th_eq3}
[\Op(a_{X_0}),\Op(a_{\xi_0})]=-\ie\dot{\lambda}(X_0)\xi_0. 
\end{equation}
Now, by using eq.~(2.10) in \cite{BB09} we get 
for every $X\in\gg$
\begin{equation}\label{ineq_th_eq4}
(\dot{\lambda}(X_0)\xi_0)(X)
=\frac{\de}{\de t}\Big\vert_{t=0}\hake{\xi_0,(-tX_0)\ast X}. 
\end{equation}
Note that the Baker-Campbell-Hausdorff formula gives that for every $t\in\RR$ 
\begin{equation}\label{ineq_th_eq5} 
(-tX_0)\ast X=-tX_0+X+t\sum_{j\ge1}b_j(\ad_{\gg}X)^jX_0+t^2P(t,X,X_0),
\end{equation}
where $\{b_j\}_{j\ge1}$ is sequence of real numbers while $P\colon\RR\times\gg\times\gg\to\gg$ 
is a certain polynomial mapping. 

On the other hand, for every $X\in\gg$ and $t\in\RR$ we have 
$\xi_0\circ\ee^{t\ad_{\gg} X}\in\Oc$ hence 
$$c_0=\hake{\xi_0\circ\ee^{t\ad_{\gg} X},X_0}
=\sum_{j\ge0}\frac{t^j}{j!}\hake{\xi_0\circ(\ad_{\gg}X)^j,X_0}. $$
Thence $\hake{\xi_0\circ(\ad_{\gg}X)^j,X_0}=0$ for every $j\ge 1$ and $X\in\gg$. 
By combining this with \eqref{ineq_th_eq3}, \eqref{ineq_th_eq4}, and \eqref{ineq_th_eq5}, 
we get \eqref{ineq_th_eq1}. 
Then the inequality \eqref{ineq_th_eq3} follows by general arguments; 
see for instance Prop.~2.1 in \cite{FS97}. 
\end{proof}

In the following statement we use the notation $\delta_{jk}$ for Kronecker's delta. 

\begin{corollary}\label{ineq_cor}
Let $\{X_1,\dots,X_n\}$ be a Jordan-H\"older basis in $\gg$ 
and denote by $\{\xi_1,\dots,\xi_n\}$ the dual basis in $\gg^*$. 
If $1\le k\le j\le n$, then we have 
\begin{equation}\label{ineq_cor_eq1}
[\Op(a_{X_j}),\Op(a_{\xi_k})]=\ie\delta_{jk}\id_{L^2(\gg)} 
\end{equation}
and 
\begin{equation}\label{ineq_cor_eq2}
\Vert\Op(a_{X_j})f\Vert\cdot\Vert\Op(a_{\xi_k})f\Vert\ge\frac{\delta_{jk}}{2}
\end{equation}
whenever $f\in L^2(\gg)$ with $\Vert f\Vert=1$ belongs to the domains of both operators 
$\Op(a_{X_j})$ and $\Op(a_{\xi_k})$. 
\end{corollary}

\begin{proof}
The hypothesis that $\{X_1,\dots,X_j\}$ is a Jordan-H\"older basis in $\gg$  
implies that for $j=1,\dots,n$ we have $[X_j,\gg]\subseteq\spa\{X_l\mid j<l\le n\}$. 
Then we can apply Theorem~\ref{ineq_th} 
to get the conclusion. 
\end{proof}

\begin{remark}\label{ineq_ex}
\normalfont
It is often the case that a coadjoint orbit of a nilpotent Lie group 
is contained in an affine subspace, like in Theorem~\ref{ineq_th}. 
Here are a few specific situations:
\begin{enumerate}
\item 
When $\xi_0\in\gg^*$ vanishes on $[\gg,\gg]$, 
its coadjoint orbit reduces to $\{\xi_0\}$, hence 
it is clearly contained in many affine subspaces. 
\item 
If $\gg$ s a two-step nilpotent Lie algebra, then 
every coadjoint orbit is an affine subspace. 
\item For every coadjoint orbit $\Oc\subseteq\gg^*$ and every $Z_0$ in the center 
of $\gg$ there exists $c_0\in\RR$ such that 
$\Oc\subseteq\{\xi\in\gg^*\mid\hake{\xi,Z_0}=c_0\}$. 
\end{enumerate}
We also note that the hypothesis $\Oc\subseteq\{\xi\in\gg^*\mid\hake{\xi,X_0}=c_0\}$ 
in Theorem~\ref{ineq_th} is equivalent to the fact that some (actually, every) 
$\xi_0\in\Oc$ vanishes on the ideal 
generated by $[X_0,\gg]$ in $\gg$. 
This implies that if $\gg$ is a two-step nilpotent Lie algebra 
in Corollary~\ref{ineq_cor} then the conclusion holds 
for every $j$ and $k$. 
\qed
\end{remark}

\section{Uncertainty principles for magnetic ambiguity functions}\label{sect3}

In this section we establish a version of Lieb's uncertainty principle (\cite{Li90}) along with some of its consequences in the present setting. 
The main result is Theorem~\ref{lieb} and is 
stated in terms of magnetic ambiguity functions and mixed-norm Lebesgue spaces 
on the cotangent bundle of a nilpotent Lie group.

\subsection{Magnetic modulation spaces}\label{subsect3.1}
We first introduce the magnetic modulation spaces on 
a simply connected nilpotent Lie group~$G$. 
The natural tool for that purpose proves to be the ambiguity function 
and not a short-time Fourier transform, 
as it is customary in the classical case 
when the nilpotent Lie group $G$ is the additive group $(\RR^n,+)$ 
(see for instance \cite{Gr01}). 
Nevertheless, our notion of modulation space 
agrees with the classical one because of 
the well-known relationship between 
the ambiguity function and the short-time Fourier transform. 

\begin{definition}\label{mod_def}
\normalfont 
Assume $1\le p,q\le\infty$ and let $\phi\in\Sc(\gg)$.  
For every tempered distribution $f\in\Sc'(\gg)$ define 
$$\Vert f\Vert_{M^{p,q}_\phi}
=\Bigl(\int\limits_{\gg}\Bigl(\int\limits_{\gg^*}\vert(\Ac_\phi f)(X,\xi)\vert^q\de\xi \Bigr)^{p/q}
\de X\Bigr)^{1/p}\in[0,\infty] $$
with the usual conventions if $p$ or $q$ is infinite. 
Then the space 
$$M^{p,q}_\phi(\gg):=\{f\in\Sc'(\gg)\mid\Vert f\Vert_{M^{p,q}}<\infty\}$$ 
will be called a \emph{magnetic modulation space} on the Lie group $G=(\gg,\ast)$. 
\qed
\end{definition}

\begin{remark}\label{mod_mixed}
\normalfont
In the setting of Definition~\ref{mod_def} let us introduce 
the \emph{mixed-norm space} $L^{p,q}(\gg\times\gg^*)$ consisting of the (equivalence classes of) 
Lebesgue measurable functions $\Theta\colon \gg\times\gg^*\to\CC$ 
such that 
$$\Vert\Theta\Vert_{L^{p,q}}
:=\Bigl(\int\limits_{\gg}\Bigl(\int\limits_{\gg^*}\vert(\Theta(x,\xi)\vert^q\de\xi \Bigr)^{p/q}
\de x\Bigr)^{1/p}<\infty $$
(cf.~\cite{Gr01}). 
It is clear that 
$M^{p,q}_\phi(\gg)=\{f\in\Sc'(\gg)\mid \Ac_\phi f\in L^{p,q}(\gg\times\gg^*)\}$. 
\qed
\end{remark}

\begin{example}\label{mod_L2}
\normalfont
For any choice of $\phi\in\Sc(\gg)$ in Definition~\ref{mod_def} we have 
$$M^{2,2}(\gg):=M^{2,2}_\phi(\gg)=L^2(\gg).$$
To see this, just note that the operator $\Ac_\phi\colon L^2(\gg)\to L^2(\gg\times\gg^*)$ 
satisfies 
$$\Vert\Ac_\phi f\Vert_{L^2(\gg\times\gg^*)}=\Vert\phi\Vert_{L^2(\gg)}\cdot\Vert f\Vert_{L^2(\gg)}$$ 
for every $f\in L^2(\gg)$, by Theorem~\ref{o6}\eqref{o6_item1}. 
Therefore 
$$\Vert f\Vert_{M^{2,2}}= \Vert\phi\Vert_{L^2(\gg)}\cdot\Vert f\Vert_{L^2(\gg)}\in[0,\infty]$$
for each $f\in\Sc'(\gg)$. 
\qed
\end{example}

\begin{notation}\label{mod_conv}
\normalfont
For every real number $p\in(1,\infty)$ we shall denote $$p':=p/(p-1)\in(1,\infty),$$ 
so that $\frac{1}{p}+\frac{1}{p'}=1$. 
\qed
\end{notation}

\subsection{Uncertainty principles for ambiguity functions}\label{subsect3.2}
In the following theorem we extend Lieb's uncertainty principle (\cite{Li90}) 
to the present setting that takes into account a magnetic potential on 
a nilpotent Lie group~$G$. 
In the special case when the magnetic potential vanishes, 
the Lie group is the abelian group $(\RR^n,+)$, and 
the estimate for $\Ac_{\phi_1}f_1\cdot\overline{\Ac_{\phi_2}f_2}$ 
is an ordinary $L^p$ one instead of a mixed-norm one, 
we recover Th.~4.1 in \cite{BDO07}, 
due to the simple relationship between the ambiguity functions and 
the short-time Fourier transforms on abelian groups.

\begin{theorem}\label{lieb} 
Let $\gg$ be a nilpotent Lie algebra with 
the corresponding simply connected Lie group $G=(\gg,\ast)$. 
\begin{enumerate} 
\item\label{lieb_item1} 
If the following conditions are satisfied: 
\begin{enumerate}
\item $p_1,p_2\in(1,\infty)$; 
\item $r_j,s_j\ge\max\{p_j,p_j'\}$ $(\ge2)$ for $j=1,2$; 
\item $p=(\frac{1}{r_1}+\frac{1}{r_2})^{-1}$ and $q=(\frac{1}{s_1}+\frac{1}{s_2})^{-1}$; 
\item $t_j=(\frac{1}{r_j}+\frac{1}{s_j'}-\frac{1}{p_j})^{-1}$ for $j=1,2$, 
\end{enumerate}
then for every $f_j\in L^{p_j}(\gg)$ and $\phi_j\in L^{t_j}(\gg)$ 
for $j=1,2$ we have 
$$\Vert\Ac_{\phi_1}f_1\cdot\overline{\Ac_{\phi_2}f_2}\Vert_{L^{p,q}(\gg\times\gg^*)}
\le C \cdot\Vert f_1\Vert_{L^{p_1}(\gg)}\cdot\Vert f_2\Vert_{L^{p_2}(\gg)}
\cdot\Vert \phi_1\Vert_{L^{t_1}(\gg)}\cdot\Vert\phi_2\Vert_{L^{t_2}(\gg)}, $$
where $C\in(0,1)$ is a certain constant depending only on 
 $p_1,p_2,r_1,r_2,s_1,s_2$, and $\dim \gg$. 
\item\label{lieb_item2} 
For every $p\ge1$ and $f_1,f_2,\phi_1,\phi_2\in L^2(\gg)$ we have 
$$\Vert\Ac_{\phi_1}f_1\cdot\overline{\Ac_{\phi_2}f_2}\Vert_{L^{p}(\gg\times\gg^*)}
\le (p^{-1/p})^{\dim \gg} \cdot\Vert f_1\Vert_{L^{2}(\gg)}\cdot\Vert f_2\Vert_{L^{2}(\gg)}
\cdot\Vert \phi_1\Vert_{L^{2}(\gg)}\cdot\Vert\phi_2\Vert_{L^{2}(\gg)}.$$
\end{enumerate}
\end{theorem}

\begin{proof} 
It is enough to prove these inequalities for $f_1,f_2,\phi_1,\phi_2\in\Sc(\gg)$. 
Note that 
\begin{equation}\label{lieb_eq1}
\Vert\Ac_{\phi_1}f_1\cdot\overline{\Ac_{\phi_2}f_2}\Vert_{L^{p,q}(\gg\times\gg^*)} 
=\Bigl(\int\limits_{\gg} 
\Vert\Ac_{\phi_1}f_1(X,\cdot)\cdot
\overline{\Ac_{\phi_2}f_2(X,\cdot)}\Vert_{L^{q}(\gg^*)}
\de X\Bigr)^{1/p}
\end{equation}
Since $\frac{1}{q}=\frac{1}{s_1}+\frac{1}{s_2}$, 
we can use H\"older's inequality to get 
\begin{equation}\label{lieb_eq2}
\Vert\Ac_{\phi_1}f_1(X,\cdot)\cdot
\overline{\Ac_{\phi_2}f_2(X,\cdot)}\Vert_{L^{q}(\gg^*)} 
\le \Vert\Ac_{\phi_1}f_1(X,\cdot)\Vert_{L^{s_1}(\gg^*)}
\cdot 
\Vert\Ac_{\phi_2}f_2(X,\cdot)\Vert_{L^{s_2}(\gg^*)} 
\end{equation}
for almost every $X\in\gg$. 
Now note that Proposition~\ref{o7} implies that 
$${\Ac}_\phi f(X,\xi)=\int\limits_{\gg}\ee^{-\ie\langle\xi,Z\rangle}
\overline{\tau_A(X,-\Psi_X^{-1}(\hskip-2pt -Z))}
f(-\Psi_X^{-1}(\hskip-2pt -Z))
\overline{\phi((-X)\ast (-\Psi_X^{-1}(\hskip-2pt -Z)))}\de Z
$$
for $f,\phi\in\Sc(\gg)$. 
Therefore, since $s_j\ge2$, we can apply the Hausdorff-Young inequality 
for the Fourier transform $L^{s_j'}(\gg)\to L^{s_j}(\gg^*)$ 
to obtain 
$$\begin{aligned}
\Vert\Ac_{\phi_j}f_j(X,\cdot)\Vert_{L^{s_j}(\gg^*)} 
\le&\Bigl(\int\limits_{\gg}\vert\overline{\tau_A(X,-\Psi_X^{-1}(-Z))}
f_j(-\Psi_X^{-1}(-Z)) \\
&\times \overline{\phi_j((-X)\ast (-\Psi_X^{-1}(-Z)))} \vert^{s_j'} \de Z\Bigr)^{1/s_j'} \\
=&\Bigl(\int\limits_{\gg}\vert
f_j(-\Psi_X^{-1}(-Z))\phi_j((-X)\ast (-\Psi_X^{-1}(-Z))) \vert^{s_j'} \de Z\Bigr)^{1/s_j'} \\
=&\Bigl(\int\limits_{\gg}\vert
f_j(Y)\phi_j((-X)\ast Y) \vert^{s_j'} \de Y\Bigr)^{1/s_j'}
\end{aligned}
$$
where we have performed the change of variables $Y=-\Psi_X^{-1}(-Z)$ 
and used the fact that $\Psi_X\colon\gg\to\gg$ is a diffeomorphism 
with the Jacobian equal to 1 everywhere on $\gg$ 
by Proposition~3.2 in~\cite{BB09}. 
If we define $\widetilde{\phi}_j(v):=\phi_j(-v)$ for every $v\in\gg$, 
it then follows that for almost every $X\in\gg$ we have 
\begin{equation}\label{lieb_eq3}
\Vert\Ac_{\phi_j}f_j(X,\cdot)\Vert_{L^{s_j}(\gg^*)} 
\le((\vert f_j\vert^{s_j'}\star\vert\widetilde{\phi}_j\vert^{s_j'})(X))^{1/s_j'},
\end{equation}
where $\star$ stands for the usual convolution product 
of functions on the nilpotent Lie group~$G$. 

On the other hand, by \eqref{lieb_eq1}~and~\eqref{lieb_eq2} we get 
\begin{align}
\Vert\Ac_{\phi_1}f_1\cdot\overline{\Ac_{\phi_2}f_2}\Vert_{L^{p,q}(\gg\times\gg^*)} 
\le&\Bigl(\int\limits_{\gg}
\Vert\Ac_{\phi_1}f_1(X,\cdot)\Vert_{L^{s_1}(\gg^*)}^{p}
\Vert\Ac_{\phi_2}f_2(X,\cdot)\Vert_{L^{s_2}(\gg^*)}^{p} \de X\Bigr)^{1/p} \nonumber\\
\le&\Bigl(\int\limits_{\gg}
\Vert\Ac_{\phi_1}f_1(X,\cdot)\Vert_{L^{s_1}(\gg^*)}^{r_1}\de X\Bigr)^{1/r_1} \nonumber\\
\label{lieb_eq4} 
&\times 
\Bigl(\int\limits_{\gg}
\Vert\Ac_{\phi_2}f_2(X,\cdot)\Vert_{L^{s_2}(\gg^*)}^{r_2} \de X\Bigr)^{1/r_2},
\end{align}
where the latter inequality follows by H\"older's inequality 
since $\frac{1}{p}=\frac{1}{r_1}+\frac{1}{r_2}$. 
Now note that by \eqref{lieb_eq3} we get 
\begin{eqnarray}
\Bigl(\int\limits_{\gg}
\Vert\Ac_{\phi_j}f_j(X,\cdot)\Vert_{L^{s_j}(\gg^*)}^{r_j} \de X\Bigr)^{1/r_j} 
&\le& \Bigl(\int\limits_{\gg}
 ((\vert f_j\vert^{s_j'}\star\vert\widetilde{\phi}_j\vert^{s_j'})(X))^{r_j/s_j'}
 \de X\Bigr)^{1/r_j} \nonumber\\
\label{lieb_eq5}&=&\Vert\vert f_j\vert^{s_j'}\star\vert\widetilde{\phi}_j\vert^{s_j'}
\Vert_{L^{r_j/s_j'}(\gg)}^{1/s_j'}.
 \end{eqnarray}
The Hausdorff inequality on 
the connected, simply connected, nilpotent Lie group $G$ 
(see Corollary~2.5' in \cite{KR78} or 
Corollary to Th.~3 in \cite{Ni94}) implies that 
for a certain constant $C_j\in(0,1)$ depending only on $p_j,r_j,s_j$, and $\dim \gg$ we have 
$$\begin{aligned}
\Vert\vert f_j\vert^{s_j'}\star\vert\widetilde{\phi}_j\vert^{s_j'}
\Vert_{L^{r_j/s_j'}(\gg)}^{1/s_j'}
& \le C_j\cdot \Vert\vert f_j\vert^{s_j'}\Vert_{L^{\alpha_j}(\gg)}^{1/s_j'}\cdot  
\Vert\vert\widetilde{\phi}_j\vert^{s_j'}
\Vert_{L^{\beta_j}(\gg)}^{1/s_j'} \\
&=C_j\cdot \Vert f_j\Vert_{L^{s_j'\alpha_j}(\gg)}\cdot  
\Vert\phi_j\Vert_{L^{s_j'\beta_j}(\gg)} \\
&=C_j\cdot \Vert f_j\Vert_{L^{p_j}(\gg)}\cdot  
\Vert\phi_j\Vert_{L^{t_j}(\gg)},
\end{aligned}
 $$
 where $\alpha_j:=p_j/s_j'$ while $\beta_j$ is chosen 
 such that $\frac{s_j'}{r_j}+1=\frac{1}{\alpha_j}+\frac{1}{\beta_j}$. 
 It is easily checked that $s_j'\beta_j=t_j$. 
 It then follows by \eqref{lieb_eq4} and \eqref{lieb_eq5} 
 that the asserted estimate holds for the constant $C:=C_1C_2$ that depends only on 
 $p_1,p_2,r_1,r_2,s_1,s_2$, and $\dim \gg$. 
 
 To prove Assertion~\eqref{lieb_item2}, recall from Corollary~2.5' in \cite{KR78} or 
Corollary to Th.~3 in \cite{Ni94} that if we denote 
 $$(\forall l\in(1,\infty))\quad A_l=\Bigl(\frac{l^{1/l}}{l'^{1/l'}}\Bigr)^{1/2},  $$
 then for $j=1,2$ we have $C_j=(A_{\alpha_j}A_{\beta_j}A_{\gamma_j})^{\dim\gg}$, 
 where $\gamma_j:=\frac{r_j/s'_j}{(r_j/s'_j)-1}$. 
 By considering the special case $p_1=p_2=2$ and $r_1=r_2=s_1=s_2=2p=2q\ge 2$, 
 a careful analysis of the constants 
 (which are the same as in the case when $\gg$ is abelian)
 then leads to the conclusion we wish for; 
 see the proof of Th.~4.1 and Cor.~4.2 in \cite{BDO07} for details.  
\end{proof}

With Theorem~\ref{lieb} at hand, one can obtain several versions 
of the uncertainty principle for the ambiguity function 
on the nilpotent Lie group $G$ in the present magnetic setting; 
see Corollaries \ref{lieb_cor2} and \ref{lieb_cor3} below.  
Before to draw these consequences, 
we note the relationship 
between the magnetic modulation spaces and the $L^p$ spaces 
on the Lie group $G$. 
We refer to \cite{GG02} for more general properties of this type 
in the case when $G$ is the abelian Lie group~$(\RR^n,+)$. 

\begin{corollary}\label{lieb_cor1}
Let $G$ be a connected, simply connected, nilpotent Lie group 
with the Lie algebra $\gg$ 
and assume that the following conditions are satisfied: 
\begin{enumerate}
\item $p\in(1,\infty)$; 
\item $r,s\ge\max\{p,p'\}$ ($\ge2$); 
\item $t=(\frac{1}{r}+\frac{1}{s'}-\frac{1}{p})^{-1}$. 
\end{enumerate}
Then for every $f\in L^{p}(\gg)$ and $\phi\in L^{t}(\gg)$ 
for $j=1,2$ we have 
$$\Vert\Ac_{\phi}f\Vert_{L^{r,s}(\gg\times\gg^*)}
\le C \cdot\Vert f\Vert_{L^{p}(\gg)}\cdot\Vert \phi\Vert_{L^{t}(\gg)}, $$
where $C\in(0,1)$ is a certain constant depending only on 
 $p,r,s$, and $\dim \gg$. 
 In particular, we have a continuous embedding 
$$L^p(\gg)\hookrightarrow M^{r,s}_\phi(\gg) \quad\text{if}\quad r,s\ge\max\{p,p'\}$$
for every $\phi\in\Sc(\gg)$. 
\end{corollary}

\begin{proof}
Just consider the special case of Theorem~\ref{lieb} with 
$p_1=p_2$, $r_1=r_2$, $s_1=s_2$, $\phi_1=\phi_2$, and $f_1=f_2$. 
\end{proof}

The next corollary is the version in the present setting 
for Th.~4.2 and Remark~4.4 in \cite{BDO07} 
or Th.~3.3.3 in \cite{Gr01}, which were stated 
in terms of the short-time Fourier transforms on~$\RR^n$.

\begin{corollary}\label{lieb_cor2}
Assume the setting of Theorem~\ref{lieb}\eqref{lieb_item1} with $r_1=s_1$ and $r_2=s_2$
and denote 
$h=(\frac{1}{\max\{p_1,p_1'\}}+\frac{1}{\max\{p_2,p_2'\}})^{-1}$.  
If the number $\varepsilon>0$ and the Borel subset $U\subseteq\gg\times\gg^*$ 
satisfy the inequality 
$$\begin{aligned}
\iint\limits_U\vert(\Ac_{\phi_1}f_1 &\cdot\overline{\Ac_{\phi_2}f_2})(X,\xi)\vert\,\de X\de\xi \\
&\ge(1-\varepsilon)\Vert f_1\Vert_{L^{p_1}(\gg)}\cdot\Vert f_2\Vert_{L^{p_2}(\gg)}
\cdot\Vert \phi_1\Vert_{L^{p_1'}(\gg)}\cdot\Vert\phi_2\Vert_{L^{p_2'}(\gg)}, 
\end{aligned}$$
then the Lebesgue measure of $U$ is at least $\sup\limits_{p>h}((1-\varepsilon)C)^{p/(p-1)}$. 
If moreover $p_1=p_2=2$, then the measure of $U$ is greater than 
$\sup\limits_{p>2}(1-\varepsilon)^{p/(p-2)}(p/2)^{(2\dim\gg)/(p-2)}$. 
\end{corollary}

\begin{proof}
Use the method of proof of Th.~4.2 in \cite{BDO07} 
or Th.~3.3.3 in \cite{Gr01}, by relying on our Theorem~\ref{lieb}. 
\end{proof}

We now record an estimate for the entropy of the ambiguity function. 
This is obtained by a method similar to the one indicated for obtaining~(6.9) in \cite{FS97}. 

\begin{corollary}\label{lieb_cor3}
Let $f,\phi\in L^2(\gg)$ such that 
$\Vert f\Vert_{L^{2}(\gg)}\cdot\Vert\phi\Vert_{L^{2}(\gg)}=1$, 
and denote 
$$\rho_{f,\phi}(\cdot):=\vert(\Ac_{\phi}f)(\cdot)\vert^2\in\bigcap\limits_{p\ge1}L^p(\gg\times\gg^*).$$
Then we have 
$$-\iint\limits_{\gg\times\gg^*}\rho_{f,\phi}\log\rho_{f,\phi}\ge \dim\gg\ge 1.$$
\end{corollary}

\begin{proof}
For every $p\ge1$ denote 
$$\gamma(p)=\iint\limits_{\gg\times\gg^*}(\rho_{f,\phi}(\cdot))^{p}
\quad\text{and } 
\chi(p)=p^{-\dim\gg}.$$ 
Then Theorem~\ref{lieb}\eqref{lieb_item2} implies that $\gamma(p)\le\chi(p)$ for every $p\ge1$. 
On the other hand, it follows at once by Proposition~\ref{o7}\eqref{o7_item1} 
that $\rho_{f,\phi}(\cdot)\le 1$ on $\gg\times\gg^*$, 
hence $\gamma(\cdot)$ is a nonincreasing function on $[1,\infty)$. 
Since so is the function $\chi(\cdot)$, and $\gamma(1)=\chi(1)$ by Theorem~\ref{o6}\eqref{o6_item1}, 
it then follows that $\gamma'(1)\le\chi'(1)$, 
which is just the inequality we wish for. 
\end{proof}

\section{The case of two-step nilpotent Lie algebras}\label{sect4}

In this section we are going to point out some specific features 
of the above constructions in the special case of 
a \emph{two-step nilpotent Lie algebra} $\gg$
(that is, $[\gg,[\gg,\gg]]=\{0\}$). 
The importance of this situation is partially motivated by 
the fact that it covers the Heisenberg algebras, 
which are characterized by the property $\dim[\gg,\gg]=1$. 
On the other hand, this class of Lie algebras 
(also known as \emph{metabelian Lie algebras}) 
contains many algebras which are neither abelian nor Heisenberg. 
In fact, the classification of two-step nilpotent Lie algebras 
is still an open problem 
although it was raised a long time ago 
(see \cite{GT99}, \cite{GK00}, 
and the references therein). 
To emphasize the richness of the class of two-step nilpotent Lie algebras, 
let us just mention that in every dimension $\ge9$ there exist 
infinitely many algebras of this type which are nonisomorphic to each other 
(see \cite{Sa83}).
By contrast, there exist precisely one abelian Lie algebra 
and at most one Heisenberg algebra in each dimension. 

\begin{example}\label{finite-dim}
\normalfont 
Here we show that nilpotent Lie algebras with arbitrarily high nilpotency index 
can be constructed as semidirect products of two-step nilpotent Lie algebras 
and appropriate function spaces thereon. 
These algebras were considered in several papers for the study 
of Schr\"odinger operators with polynomial magnetic fields; 
see for instance \cite{JK85} and \cite{BL06} and the references therein. 

Let $\gg$ be a two-step nilpotent Lie algebra 
and $N\ge1$ a fixed integer. 
Denote by $\Pc_N(\gg)$ the finite-dimensional linear space 
of real polynomial functions of degree~$\le N$ on $\gg$. 
Then $\Fc:=\Pc_N(\gg)$ is an admissible function space 
in the sense of Def.~2.8 in \cite{BB09} 
(see also  Setting~\ref{o1} above). 
Note that if we think of $\gg$ as a Lie group 
with respect to the Baker-Campbell-Hausdorff multiplication 
$$(\forall X,Y\in\gg)\quad X\ast Y=X+Y+\frac{1}{2}[X,Y], $$
then 
$\Pc_N(\gg)$ is invariant under the left translations on $\gg$ 
since every left translation $Y\mapsto X\ast Y$ 
is a polynomial mapping of degree $\le 1$. 

By using the formula for the bracket in the semidirect 
product of Lie algebras $\mg=\Fc\rtimes_{\dot\lambda}\gg$, 
$$[(f_1,X_1),(f_2,X_2)]=(\dot\lambda(X_1)f_2-\dot\lambda(X_2)f_1,[X_1,X_2]) $$
it is easy to see that $\mg$ is a nilpotent Lie algebra 
whose nilpotency index is at least~$\max\{N,2\}$. 
It also follows that the center of $\mg$ is 
$\Pc_N^0(\gg)\times\zg$
where $\zg$ is the center of $\gg$ and 
$$\Pc_N^0(\gg)=\{f\in\Pc_N(\gg)\mid f'=0\}=\{f\in\Pc_N(\gg)\mid f'_0=0\}. $$
Here we have denoted by $f'_Y\in\gg^*$ the differential of $f$ at some point $Y\in\gg$, 
while 
$$(\dot\lambda(X)f)(Y)=\frac{\de}{\de t}\Big\vert_{t=0}f((-tX)\ast Y)=f'_Y(-X-\frac{1}{2}[X,Y])$$
(compare formula~(2.10) in \cite{BB09}).
\qed 
\end{example}

\begin{example}\label{part}
\normalfont 
Let $\gg$ be a two-step nilpotent Lie algebra again 
and denote the center of $\gg$ by $\zg$. 
It follows by Example~\ref{pseudo_ex}\eqref{pseudo_ex_item2} 
that for every $X_0\in\gg$ the corresponding right-invariant vector field on $\gg$ 
is 
$$\overline{X}_0\colon\gg\to\gg,\quad  \overline{X}_0(Y)=X_0-\frac{1}{2}[Y,X_0].$$
(In particular, if $X_0\in\zg$, then $\overline{X}_0$ defines 
a first-order differential operator $\dot{\lambda}(X_0)$ with constant coefficients 
in any coordinate system on $\gg$.)

On the other hand, for every $\xi_0\in\gg^*$ we have 
$$(\forall Y\in\gg)\quad (\dot{\lambda}(X_0)\xi_0)(Y)=\hake{\xi_0,\overline{X}_0(Y)}
=\hake{\xi_0,X_0}-\frac{1}{2}\hake{\xi_0,[Y,X_0]}. $$
Since $[\gg,[\gg,\gg]]=\{0\}$, it follows that $[X_0,\gg]$ is an ideal in $\gg$. 
\qed
\end{example}

\begin{corollary}\label{wig_2}
Let $\gg$ be a two-step nilpotent Lie algebra and $f,\phi\in\Sc(\gg)$ arbitrary. 
\begin{enumerate}
\item 
For every $(X,\xi)\in\gg\times\gg^*$ we have 
$$(\Ac_\phi f)(X,\xi)=\int\limits_{\gg}\ee^{-\ie\hake{\xi,Y}} 
\cdot\overline{\tau_A(X,(X/2)\ast Y)}\cdot
f((X/2)\ast Y)\cdot \overline{\phi((-X/2)\ast Y)} \de Y.$$
\item 
For every $(Y,\eta)\in\gg\times\gg^*$ we have 
$$\Wig(f,\phi)(Y,\eta)=\int\limits_{\gg}\ee^{-\ie\hake{\eta,X}}
\cdot\overline{\tau_A(X,(X/2)\ast Y)}\cdot
f((X/2)\ast Y)\cdot\overline{\phi((-X/2)\ast Y)} \de X.$$ 
\end{enumerate}
\end{corollary}

\begin{proof}
Since $[\gg,[\gg,\gg]]=\{0\}$, it follows at once that 
for every $X,Y\in\gg$ we have 
$\Psi_X(Y)=Y\ast(X/2)$ 
and 
$\tau_A(X,Y)=\exp\Bigl(\ie\int\limits_0^1\langle A_{(-sX)\ast Y},X\rangle\de s\Bigr)$. 
Then use Proposition~\ref{o7}. 
\end{proof}

\subsection*{Acknowledgment} 
We wish to thank Professor Jos\'e Gal\'e for his kind help.

\end{document}